\documentclass{svproc}
\pdfoutput=1

\usepackage{bbm}
\usepackage{amssymb}
\usepackage{mathtools}
\usepackage{marvosym}
\usepackage{algorithm}
\usepackage[noend]{algpseudocode}
\usepackage{hyperref}
\usepackage[capitalize]{cleveref}

\DeclarePairedDelimiter{\parens}{\lparen}{\rparen}

\DeclareMathOperator{\prange}{range}
\DeclareMathOperator*{\argmin}{arg\,min}
\DeclareMathOperator{\pcost}{cost}
\DeclareMathOperator{\pexpected}{E}
\DeclareMathOperator{\pball}{B}
\DeclareMathOperator{\perr}{rerror}

\newcommand{\expected}[1]{\ensuremath{\pexpected\left[#1\right]}}
\newcommand{\confer}{cf.~}
\newcommand{\range}[3]{\ensuremath{\prange\parens*{#1,#2,#3}}}
\newcommand{\ind}[1]{\ensuremath{\mathbbm{1}\parens*{#1}}}
\newcommand{\serror}[2]{\ensuremath{\perr\parens*{#1,#2}}}

\newcommand{\cost}[2]{\ensuremath{\pcost\parens*{#1,#2}}}
\newcommand{\scost}[2]{\ensuremath{\widehat{\pcost}\parens*{#1,#2}}}
\newcommand{\ucost}[2]{\ensuremath{\widetilde{\pcost}\parens*{#1,#2}}}
\let\epsilon\relax
\newcommand{\epsilon}{\varepsilon}

\allowdisplaybreaks

\begin{document}
\mainmatter 

\title{Coresets for $(k, \ell)$-Median Clustering under the Fréchet Distance}

\author{Maike Buchin\textsuperscript{
\href{https://orcid.org/0000-0002-3446-4343}{[0000-0002-3446-4343]}} \and Dennis Rohde\textsuperscript{
\href{https://orcid.org/0000-0001-8984-1962}{[0000-0001-8984-1962]} \Letter}}
\institute{Faculty of Informatics, Ruhr-University Bochum\\ \email{maike.buchin@rub.de} \\ \email{dennis.rohde-t1b@rub.de}}

\authorrunning{M. Buchin and D. Rohde}
\maketitle

\begin{abstract}
    We present an algorithm for computing $\epsilon$-coresets for $(k, \ell)$-median clustering of polygonal curves in $\mathbb{R}^d$ under the Fr\'echet distance. This type of clustering is an adaption of Euclidean $k$-median clustering: we are given a set of $n$ polygonal curves in $\mathbb{R}^d$, each of complexity (number of vertices) at most $m$, and want to compute $k$ median curves such that the sum of distances from the given curves to their closest median curve is minimal. Additionally, we restrict the complexity of the median curves to be at most $\ell$ each, to suppress overfitting, a problem specific for sequential data. Our algorithm has running time linear in $n$, sub-quartic in $m$ and quadratic in $\epsilon^{-1}$. With high probability it returns $\epsilon$-coresets of size quadratic in $\epsilon^{-1}$ and logarithmic in $n$ and $m$. We achieve this result by applying the improved $\epsilon$-coreset framework by Langberg and Feldman to a generalized $k$-median problem over an arbitrary metric space. Later we combine this result with the recent result by Driemel et al. on the VC dimension of metric balls under the Fr\'echet distance. Furthermore, our framework yields $\epsilon$-coresets for any generalized $k$-median problem where the range space induced by the open metric balls of the underlying space has bounded VC dimension, which is of independent interest. Finally, we show that our $\epsilon$-coresets can be used to improve the running time of an existing approximation algorithm for $(1,\ell)$-median clustering.%, taking a further step in the direction of practical $(k,\ell)$-median clustering algorithms.
\end{abstract}

\section{Introduction}

At the present time even efficient approximation algorithms are often incapable of handling massive data sets, which have become common. Here, we need efficient methods to reduce data while (approximately) maintaining the core properties of the data. A popular approach to this topic are $\epsilon$-coresets; see for example \cite{DBLP:journals/ki/MunteanuS18,DBLP:journals/corr/abs-2011-09384} for comprehensive surveys. An $\epsilon$-coreset is a small (weighted) set that aggregates certain properties of a given (massive) data set up to some small error. $\epsilon$-coresets are very popular in the field of clustering, \confer \cite{DBLP:conf/stoc/Har-PeledM04,DBLP:journals/dcg/Har-PeledK07,chen_coresets_2009,DBLP:conf/focs/HuangJLW18} and they are becoming a topic in other fields, too%, \confer \cite{DBLP:conf/nips/HugginsCB16,DBLP:conf/nips/MunteanuSSW18,DBLP:conf/iclr/SenerS18,DBLP:journals/siamcomp/FeldmanSS20,DBLP:journals/jsac/LuLHWNC20}.
, \confer \cite{DBLP:conf/nips/MunteanuSSW18,DBLP:journals/siamcomp/FeldmanSS20}.
The technique for computing an $\epsilon$-coreset for a given data set highly depends on the application at hand, but mostly $\epsilon$-coresets are computed by filtering the given data set.

While $\epsilon$-coresets can be computed efficiently for $k$-clustering of points in the Euclidean space, less is known for clustering of curves. In particular, only little effort (\confer \cite{coresets_eurocg}) has been made in designing methods for computing $\epsilon$-coresets for $(k,\ell)$-clustering of polygonal curves in $\mathbb{R}^d$ under the Fréchet distance. This type of clustering, which has recently drawn increasing popularity due to a growing number of applications, see \cite{DBLP:conf/soda/DriemelKS16,k_l_center,doi:10.1137/1.9781611976465.160,DBLP:conf/gis/BuchinDLN19}, is an adaption of the Euclidean $k$-clustering: we are given a set of polygonal curves and seek to compute $k$ center curves that minimize either the maximum Fr\'echet distance (center objective), or the sum of Fr\'echet distances (median objective), among the given curves and their closest center curve. In addition, we restrict the complexity---the number of vertices---of each center curve to be at most $\ell$ to suppress overfitting, a problem specific for sequential data. This means that input curves and center curves are in general of different complexities, which is the reason why we need a specialized algorithm for computing $\epsilon$-coresets and can not apply $\epsilon$-coreset algorithms for discrete metric spaces (\confer \cite{Feldman2011}) on the input.

The Fr\'echet distance is a natural dissimilarity measure for curves that is a pseudo-metric and can be computed efficiently \cite{alt_godau}. Unlike other measures for curves, like the dynamic time warping distance (or the discrete version of the Fr\'echet distance), it %this distance measure
takes the whole course of the curves into account, not only the pairwise distances among their vertices. This can be particularly useful, e.g. when the input consists of irregularly sampled trajectories, \confer \cite{DBLP:conf/soda/DriemelKS16}. Unfortunately, since the Fr\'echet distance is a bottleneck distance measure, i.e., it boils down to a single distance between two points on the curves, it is sensitive to outliers, which may negatively affect its applications. In clustering, we can counteract by choosing an appropriate clustering objective and indeed, the $(k,\ell)$-median objective is a good choice, because the median is a robust measure of central tendency. However, the state of the art $(k,\ell)$-median clustering algorithms (\confer \cite{DBLP:conf/soda/DriemelKS16,doi:10.1137/1.9781611976465.160}) have exponential running time dependencies and cannot be used in practice, while the practical algorithms for $(k,\ell)$-clustering (\confer \cite{k_l_center,DBLP:conf/gis/BuchinDLN19}) rely on the $(k,\ell)$-center objective, which is not robust and therefore amplifies the sensitivity on outliers.

In this work, we present an algorithm for computing $\epsilon$-coresets for $(k,\ell)$-median clustering under the Fr\'echet distance and improve an $(1,\ell)$-median clustering algorithm by Buchin et al. \cite{arxiv_full}, using $\epsilon$-coresets and rendering it much more practical.

\subsection{Related Work}

$(k, \ell)$-clustering of polygonal curves was introduced by Driemel et al.~\cite{DBLP:conf/soda/DriemelKS16}. They developed the first approximation schemes for $(k,\ell)$-center and $(k,\ell)$-median clustering of polygonal curves in $\mathbb{R}$, which run in near-linear time. They proved that both problems are NP-hard, when $k$ is part of the input. Further, they showed that the doubling dimension of the space of polygonal curves under the Fr\'echet distance is unbounded, even when the curves are of bounded complexity. Subsequently, Buchin et al.~\cite{k_l_center} presented a constant factor approximation algorithm for $(k, \ell)$-center clustering of polygonal curves in $\mathbb{R}^d$, with running time linear in the number of given curves and polynomial in their maximum complexity. Also they showed that $(k,\ell)$-center clustering is NP-hard and NP-hard to approximate within a factor of $(1.5 - \epsilon)$ for curves in $\mathbb{R}$, respectively $(2.25 - \epsilon)$ for curves in $\mathbb{R}^d$ with $d \geq 2$, even for $k = 1$. Buchin et al.~\cite{DBLP:conf/gis/BuchinDLN19} provided practical algorithms for $(k,\ell)$ clustering under the Fr\'echet distance and thereby introduce a new technique, the so called Fr\'echet centering, for computing better cluster centers. Also, Meintrup et al.~\cite{DBLP:conf/nips/MeintrupMR19} provided a practical $(1+\epsilon)$-approximation algorithm for discrete $(k,\ell)$-median clustering under the presence of a certain number of outliers. 
Buchin et al.~\cite{DBLP:conf/swat/BuchinDS20} proved that $(k,\ell)$-median clustering is also NP-hard, even for $k = 1$. Furthermore, they presented polynomial-time approximation schemes for $(k,\ell)$-center and $(k, \ell)$-median clustering of polygonal curves under the discrete Fr\'echet distance. Nath and Taylor \cite{abhin2020kmedian} gave a near-linear time approximation scheme for $(k,\ell)$-median clustering of polygonal curves in $\mathbb{R}^d$ under the discrete Fr\'echet distance and a polynomial-time approximation scheme for $k$-median clustering of sets of points from $\mathbb{R}^d$ under the Hausdorff distance. Furthermore, they showed that $k$-median clustering of point sets under the Hausdorff distance is NP-hard (for constant $k$). 
Recently, Buchin et al.~\cite{doi:10.1137/1.9781611976465.160} developed an approximation scheme for $(k,\ell)$-median clustering under the Fr\'echet distance with running time linear in the number of curves and polynomial in their complexity, where the computed centers have complexity up to $2 \ell - 2$.

Langberg and Schulman~\cite{LangbergS10} developed a framework for computing relative error approximations of integrals over any function from a given family of unbounded and non-negative real functions. In particular, this framework can be used to compute $\epsilon$-coresets for $k$-clustering of points in $\mathbb{R}^d$ with objective functions based on sums of distances among the points and their closest center. The idea of their framework is to sub-sample the input with respect to a certain non-uniform probability distribution, which is computed using an approximate solution to the problem. More precisely, the approximate solution is used to compute an upper bound on the sensitivity of each data element. The sensitivity is the maximum fraction of cost that the element may cause for any possible solution. It is a notion of the data elements \textit{importance} for the problem and the probability distribution is set up such that each element has probability proportional to its importance. A sample of a certain size drawn from this distribution and properly weighted, is an $\epsilon$-coreset for the underlying clustering problem with high probability. Feldman and Langberg \cite{Feldman2011} developed a unified framework for approximate clustering, which is largely based on $\epsilon$-coresets. They combine the techniques by Langberg and Schulman \cite{LangbergS10} with $\epsilon$-approximations, which stem from the framework of range spaces and VC dimension developed in statistical learning theory. As a result, they address a spectrum of clustering problems, such as $k$-median clustering, $k$-line median clustering, projective clustering and also other problems like subspace approximation. Braverman et al. \cite{DBLP:journals/corr/BravermanFL16} improved the aforementioned framework by switching to $(\epsilon, \eta)$-approximations, which leads to substantially smaller sample sizes in many cases. Also, they simplified and further generalized the framework and applied it to $k$-means clustering of points in $\mathbb{R}^d$. Following, Feldman et al. \cite{DBLP:journals/siamcomp/FeldmanSS20} improved this framework by switching to another range space, thereby obtaining smaller coresets for $k$-means, $k$-line means and affine subspace clustering.

\subsection{Our Contributions}

In this work we develop an algorithm for computing $\epsilon$-coresets for $(k,\ell)$-median clustering of polygonal curves in $\mathbb{R}^d$ under the Fr\'echet distance (where in the following we assume $k,\ell$ and $d$ to be constant):

\begin{theorem}\label{theo:kl_coreset1}
    There exists an algorithm that, given a set of $n$ polygonal curves in $\mathbb{R}^d$ of complexity at most $m$ each and a parameter $\epsilon \in (0,1)$, returns with constant positive probability an $\epsilon$-coreset for $(k, \ell)$-median clustering under the Fr\'echet distance of size $O(k^2 \log^2(k) \epsilon^{-2} \log(m) \log(kn))$, in time $$O(nm \log(m) + nm^3 \log(m) + \epsilon^{-2}\log(m)\log(n))$$ for $k > 1$ and $$O(n m \log(m) + m^3 \log(m) + \epsilon^{-2}\log(m)\log(n))$$ for $k = 1$.
\end{theorem}

Also we show that $\epsilon$-coresets can be used to improve the running time of an existing $(1,\ell)$-median $(5+\epsilon)$-approximation algorithm \cite{arxiv_full}, thereby facilitating its application in practice. 

We start by defining \emph{generalized $k$-median clustering}, where input and centers come from a subset (not necessarily the same) of an underlying metric space, each, and then derive our $\epsilon$-coreset result in this setting. This notion captures $(k,\ell)$-median clustering under the Fr\'echet distance in particular, but the analysis holds for any metric space. In doing so, we first give a universal bound on the so called sensitivity of the elements of the given data set and their total sensitivity, i.e., the sum of their sensitivities. The sensitivities are a measure of the data elements importance, i.e., the maximum fraction of the cost an element might cause for any center set, and later they determine the sample probabilities. Our analysis is based on the analysis of Langberg and Schulman \cite{LangbergS10}. 

Next, we apply the improved $\epsilon$-coreset framework by Feldman and Langberg \cite{Feldman2011}. Here, our analysis is based on the analysis of Feldman et al. \cite{DBLP:journals/siamcomp/FeldmanSS20}, but our sample size depends on the VC-dimension of the range space induced by the open metric balls. The open metric balls form a basis of the metric topology, hence it is more natural to study the VC dimension of their associated range space in a geometric setting. Indeed, for the $\ell_p^d$ spaces these range spaces have already been studied \cite[Theorem 2.2]{DBLP:journals/ml/GoldbergJ95} and recently, results for the (continuous and discrete) Fr\'echet, weak Fr\'echet and Hausdorff distance were obtained \cite{Driemel19}, enabling our main result. Finally, we show how an existing $(1,\ell)$-median $(5+\epsilon)$-approximation algorithm \cite{arxiv_full} can be improved by means of our $\epsilon$-coresets. 

\begin{theorem}\label{theo:approx2}
    There exists an algorithm that, given a set $T$ of $n$ polygonal curves of complexity at most $m$ each, and a parameter $\epsilon \in (0, 1/2]$, computes a polygonal curve $c$ of complexity $2 \ell - 2$, such that with constant positive probability, it holds that
    \begin{align*}
        \cost{T}{\{c\}} = \sum_{\tau \in T} d_F(\tau, c) \leq (5+\epsilon) \sum_{\tau \in T} d_F(\tau, c^\ast) = (5+\epsilon) \cost{T}{\{c^\ast\}},
    \end{align*}
    where $c^\ast$ is an optimal $(1, \ell)$-median for $T$ under the Fr\'echet distance. The algorithm has running time $$O\left(n m \log(m) + m^2 \log(m) + m^{2\ell-1}\epsilon^{-2\ell d + 2d - 2}\log^2(m)\log(n)\right).$$
\end{theorem}

\cref{theo:kl_coreset1,theo:approx2} will follow from \cref{theo:kl_coreset,theo:approx}, respectively. Note that although we do not present algorithms for computing $\epsilon$-coresets for the weak and discrete Fr\'echet and Hausdorff distance, our results also imply the existence of $\epsilon$-coresets of similar size for these metrics.

\subsection{Organization}

In \cref{sec:coresets_metric} we give the results for general metric spaces: we derive a universal bound on the sensitivities in \cref{subsec:sensitivity} and the $\epsilon$-coreset result in \cref{subsec:sensitivity_sampling}. In \cref{sec:coresets_frechet} we present the algorithm for computing $\epsilon$-coresets for $(k,\ell)$-median clustering. Finally, in \cref{sec:approximation_improvement} we demonstrate the use of $\epsilon$-coresets in an existing $(5+\epsilon)$-approximation algorithm for $(1,\ell)$-median clustering. 

\section{Coresets for Generalized $k$-Median Clustering in Metric Spaces}
\label{sec:coresets_metric}
In this section, we first derive general results for $\epsilon$-coreset based on the sensitivity sampling framework \cite{LangbergS10,Feldman2011}. In the following $d \in \mathbb{N}$ is an arbitrary constant. By $\lVert \cdot \rVert$ we denote the Euclidean norm and for $n \in \mathbb{N}$, we define $[n] = \{1, \dots, n\}$. For a closed logical formula $\Psi$ we define by $\mathbbm{1}(\Psi)$ the function that is $1$ if $\Psi$ is true and $0$ otherwise.

Let $\mathcal{X} = (X,\rho)$ be an arbitrary metric space, where $X$ is any non-empty set and $\rho \colon X \times X \rightarrow \mathbb{R}_{\geq 0}$ is a distance function. We introduce a generalized definition of $k$-median clustering, where the input is restricted to come from a predefined subset $Y \subseteq X$ and the medians are restricted to come from a predefined subset $Z \subseteq X$.

\begin{definition}
    \label{def:generalized_k_median}
    The generalized $k$-median clustering problem is defined as follows, where $k \in \mathbb{N}$ is a fixed (constant) parameter of the problem: given a finite and non-empty set $T = \{ \tau_1, \dots, \tau_n \} \subseteq Y$, compute a set $C$ of $k$ elements from $Z$, such that $\cost{T}{C} = \sum_{\tau \in T} \min_{c \in C} \rho(\tau,c)$ is minimal.
\end{definition}

We analyze the problem in terms of functions. This allows us to apply the improved $\epsilon$-coreset framework by Feldman and Langberg \cite{Feldman2011}. Therefore, given a set $T = \{ \tau_1, \dots, \tau_n \} \subseteq Y$ we define $F = \{ f_1, \dots, f_n \}$ to be a set of functions with $f_i \colon 2^{Z} \setminus \{ \emptyset \} \rightarrow \mathbb{R}_{\geq 0},\ C \mapsto \min_{c \in C} \rho(c, \tau_i)$. For each $C \in 2^Z \setminus \{ \emptyset \}$ we now have $\cost{T}{C} = \sum_{i=1}^n f_i(C)$.

In the following, we bound the sensitivity of each $\tau \in T$. That is the maximum fraction of $\cost{T}{C}$ that is caused by $\tau$, for all $C$. To comply with the $k$-median problem we only take into account the $k$-subsets $C \subseteq Z$.

\subsection{Sensitivity Bound}
\label{subsec:sensitivity}

First, we formally define the sensitivities of the inputs $\tau \in T$ in terms of the respective functions. 

\begin{definition}[\cite{Feldman2011}]
    \label{def:sensitivity}
    Let $F$ be a finite and non-empty set of functions $f \colon 2^Z \setminus \{ \emptyset \} \rightarrow \mathbb{R}_{\geq 0}$. For $f \in F$ we define the sensitivity with respect to $F$: \[ \mathfrak{s}(f,F) = \sup_{\substack{C = \{c_1, \dots, c_k\} \subseteq Z \\ \sum\limits_{g \in F} g(C) > 0}} \ \frac{f(C)}{\sum\limits_{g \in F} g(C)}. \] 
    
    We define the total sensitivity of $F$ as $\mathfrak{S}(F) = \sum_{f \in F} \mathfrak{s}(f,F)$.
\end{definition}

We now prove a bound on the sensitivity of all $f \in F$, which then yields a bound on the total sensitivity of $F$. Later, our coreset will be a weighted sample from a distribution whose probabilities are determined by the derived bounds. To compute the bounds, any (bi-criteria) approximate solution to the generalized $k$-median problem can be used. Our analysis is an adaption of the analysis of the sensitivities for sum-based $k$-clustering of points in $\mathbb{R}^d$ by Langberg and Schulman \cite{LangbergS10}. We note that similar bounds have already been derived in the literature, see e.g., \cite{DBLP:conf/fsttcs/VaradarajanX12}.

\begin{lemma}
    \label{lem:sensitivities}
    Let $k^\prime \in \mathbb{N}$, $C^\ast = \{ c^\ast_1, \dots, c^\ast_k \} \subseteq Z$ with $\Delta^\ast = \sum_{i=1}^n f_i(C^\ast)$ minimal and $\hat{C} = \{\hat{c}_1, \dots, \hat{c}_{k^\prime}\} \subseteq X$ with $\hat{\Delta} = \sum_{i=1}^n f(\hat{C}) \leq \alpha \cdot \Delta^\ast$ for an $\alpha \in [1, \infty)$. Breaking ties arbitrarily, we assume that every $\tau \in T$ has a unique nearest neighbor in $\hat{C}$ and for $i \in [k^\prime]$, we define $\hat{V}_i = \{ \tau \in T \mid \forall j \in [k^\prime]: \rho(\tau, \hat{c}_i) \leq \rho(\tau, \hat{c}_j) \}$ to be the Voronoi cell of $\hat{c}_i$ and $\hat{\Delta}_i = \sum_{\tau \in \hat{V}_i} \rho(\tau, \hat{c}_i)$ to be its cost. For each $i \in [k^\prime]$ and $\tau_j \in \hat{V}_i$ it holds that $$ \gamma(f_j) = \left(1 + \sqrt{\frac{2k^\prime}{3\alpha}} \right) \left( \frac{\alpha \rho(\tau_j, \hat{c}_i)}{\hat{\Delta}} + \frac{2 \alpha \hat{\Delta}_i}{\hat{\Delta}\lvert \hat{V}_i \rvert} \right) + \left(1 + \sqrt{\frac{3\alpha}{2k^\prime}}\right) \frac{2}{\lvert \hat{V}_i \rvert} \geq \mathfrak{s}(f_j, F)$$ and $ \Gamma = \sum_{f \in F} \gamma(f) = 2k^\prime + 2\sqrt{6\alpha k^\prime} + 3 \alpha \geq \mathfrak{S}(F)$.
\end{lemma}
\begin{proof}
    We assume that $\hat{\Delta} > 0$. By assumption $\hat{V}_1, \dots, \hat{V}_{k^\prime}$ form a partition of $T$ and by definition $\hat{\Delta} = \sum_{i=1}^{k^\prime} \hat{\Delta}_i$ as well as $\sum_{f \in F} f(C) \geq \hat{\Delta}/\alpha$ for each $C = \{ c_1, \dots, c_k\} \subseteq Z$. For $i \in [k^\prime]$, we let $\hat{B}_i = \{ \tau \in \hat{V}_i \mid \rho(\tau, \hat{c}_i) \leq 2\hat{\Delta}_i/\lvert \hat{V}_i \rvert\}$. It holds that $\lvert \hat{B}_i \rvert \geq \lvert \hat{V}_i \rvert/2$, since otherwise $\sum_{\tau \in \hat{V}_i \setminus \hat{B}_i} \rho(\tau, \hat{c}_i) > \hat{\Delta}_i$, which is a contradiction. By the triangle-inequality, we have for each $\{ c_1, \dots, c_k \} \subseteq Z$, $i \in [k^\prime]$, $j \in [k]$ and $\tau \in T$: $$ \rho(\hat{c}_i, c_j) \leq \rho(\hat{c}_i, \tau) + \rho(\tau, c_j) \iff \rho(\tau, c_j) \geq \rho(\hat{c}_i, c_j) - \rho(\hat{c}_i, \tau).$$ Furthermore, since $\rho$ is non-negative: $ \rho(\tau, c_j) \geq \max\{ 0, \rho(\hat{c}_i, c_j) - \rho(\hat{c}_i, \tau) \}$.
    
    For each $C = \{c_1, \dots, c_k\} \subseteq Z$, $i \in [k^\prime]$ and $\beta \in [0,1]$ we now have the following bound: 
    \begin{align*}
        \sum_{f \in F} f(C) & \geq \beta \sum_{\tau \in \hat{B}_i} \min_{j \in [k]} \rho(\tau, c_j) + (1-\beta) \frac{\hat{\Delta}}{\alpha} \\
        & \geq \beta \max\left\{0, \rho(\hat{c}_i, c_j) - \frac{2\hat{\Delta}_i}{\lvert \hat{V}_i \rvert} \right\} \frac{\lvert \hat{V}_i \rvert}{2} + (1 - \beta) \frac{\hat{\Delta}}{\alpha},    
    \end{align*}
    where among $\{c_1, \dots, c_k\}$, $c_j$ is closest to $\hat{c}_i$. Using the triangle-inequality and the above bound yields for each $\beta \in [0,1)$, $i \in [k^\prime]$ and $\tau_m \in \hat{V}_i$:
    \begin{align*}
        \mathfrak{s}(f_m, F) & \leq \sup_{\{c_1, \dots, c_k\} \subseteq Z} \ \frac{\rho(\tau_m, \hat{c}_i) + \rho(\hat{c}_i, c_j)}{\beta \max\left\{0, \rho(\hat{c}_i, c_j) - \frac{2\hat{\Delta}_i}{\lvert \hat{V}_i \rvert} \right\} \frac{\lvert \hat{V}_i \rvert}{2} + (1 - \beta) \frac{\hat{\Delta}}{\alpha}} \\
        & \leq \sup_{\substack{\{c_1, \dots, c_k\} \subseteq Z \\ \rho(\hat{c}_i, c_j) \geq 2\hat{\Delta}_i/\lvert \hat{V}_i \rvert}} \ \frac{\rho(\tau_m, \hat{c}_i) + \rho(\hat{c}_i, c_j)}{\beta \left(\frac{\lvert \hat{V}_i \rvert \rho(\hat{c}_i, c_j)}{2} - \hat{\Delta}_i\right) + (1 - \beta) \frac{\hat{\Delta}}{\alpha}}
    \end{align*}
    Here, again $c_j$ is closest to $\hat{c}_i$ among $\{c_1, \dots, c_k\}$ and the last inequality follows because it can be observed that the term takes smaller values for $\rho(\hat{c}_i, c_j) < 2 \hat{\Delta}_i/\lvert \hat{V}_i \rvert$ than for $\rho(\hat{c}_i, c_j) \geq 2 \hat{\Delta}_i/\lvert \hat{V}_i \rvert$, independent of $\beta$. Now, to obtain a bound that is independent of $c_j$, we substitute $\rho(\hat{c}_i, c_j)$ by a free variable $x$ and let $$ h \colon [2\hat{\Delta}_i/\lvert \hat{V}_i \rvert, \infty) \rightarrow \mathbb{R}_{\geq 0},\ x \mapsto \frac{\rho(\tau_m, \hat{c}_i) + x}{\beta \left(\frac{\lvert \hat{V}_i \rvert x}{2} - \hat{\Delta}_i \right) + (1 - \beta) \frac{\hat{\Delta}}{\alpha}}.$$ The derivative of $h$ is $$ \frac{(1-\beta)\frac{\hat{\Delta}}{\alpha} - \beta \left(\frac{\lvert \hat{V}_i \rvert \rho(\tau_m, \hat{c}_i)}{2}+\hat{\Delta}_i\right)}{\left(\beta \left(\frac{\lvert \hat{V}_i \rvert x}{2} - \hat{\Delta}_i \right) + (1 - \beta) \frac{\hat{\Delta}}{\alpha}\right)^2} $$ and it can be observed that the sign of this function is independent of $x$. Therefore, $h$ is a monotonic function and is thus either maximized at $x = 2\hat{\Delta}_i/\lvert \hat{V}_i \rvert$ or when $x \rightarrow \infty$. Using l'Hôspital's rule we obtain
    \begin{align*}
        \mathfrak{s}(f_m, F) & \leq \max \left\{\frac{\rho(\tau_m, \hat{c}_i) + \frac{2\hat{\Delta}_i}{\lvert \hat{V}_i \rvert}}{(1-\beta)\frac{\hat{\Delta}}{\alpha}}, \frac{1}{\beta \frac{\lvert \hat{V}_i \rvert}{2}} \right\} \leq \frac{\alpha \rho(\tau_m, \hat{c}_i)}{(1-\beta) \hat{\Delta}} + \frac{2 \alpha \hat{\Delta}_i}{(1-\beta) \hat{\Delta} \lvert \hat{V}_i \rvert} + \frac{2}{\beta \lvert \hat{V}_i \rvert}.
    \end{align*}
    Therefore,
    \begin{align*}
        \mathfrak{S}(F) & \leq \sum_{i=1}^{k^\prime} \sum_{\tau_m \in \hat{V}_i} \left( \frac{\alpha \rho(\tau_m, \hat{c}_i)}{(1-\beta) \hat{\Delta}} + \frac{2 \alpha \hat{\Delta}_i}{(1-\beta) \hat{\Delta} \lvert \hat{V}_i \rvert} + \frac{2}{\beta \lvert \hat{V}_i \rvert} \right) = \frac{3 \alpha}{1-\beta} + \frac{2k^\prime}{\beta}.
    \end{align*}
    By simple calculus, this bound is minimized at $\beta = \frac{1}{1+ \sqrt{\frac{3 \alpha}{2 k^\prime}}} < 1$.
\qed\end{proof}

\subsection{Coresets by Sensitivity Sampling}
\label{subsec:sensitivity_sampling}

We apply the framework of Feldman and Langberg \cite{Feldman2011}. First, we formally define $\epsilon$-coresets for generalized $k$-median clustering.

\begin{definition}
    \label{def:coreset}
    Given $\epsilon \in (0,1)$ and a finite non-empty set $T \subseteq Y$, a (multi-)set $S \subseteq X$ together with a weight function $w \colon S \rightarrow \mathbb{R}_{> 0}$ is a weighted $\epsilon$-coreset for $k$-median clustering of $T$, if for all $C \subseteq Z$ with $\lvert C \rvert = k$ it holds that $$(1-\epsilon) \cost{T}{C} \leq \pcost_w\parens*{S,C} \leq (1+\epsilon) \cost{T}{C},$$ where $\pcost_w\parens*{S,C} = \sum_{s \in S} w(s) \cdot \min_{c \in C} \rho(s, c)$. 
\end{definition}

We define range spaces and the associated concepts.

\begin{definition}
    A range space is a pair $(X,\mathcal{R})$, where $X$ is a set, called ground set and $\mathcal{R}$ is a set of subsets $R \subseteq X$, called ranges.
\end{definition}

The projection of a range space $(X, \mathcal{R})$ onto a subset $Y \subseteq X$ is the range space $(Y, \{ Y \cap R \mid R \in \mathcal{R} \})$. Furthermore, for each range space there exists a complementary range space.

\begin{definition}
    \label{def:complementary_range_space}
    Let $F = (X,\mathcal{R})$ be a range space. We call  $\overline{F} = (X, \overline{\mathcal{R}})$, the range space over $\overline{\mathcal{R}} = \{ X \setminus R \mid R \in \mathcal{R} \}$, the complementary range space of $F$.
\end{definition}

A measure of the combinatorial complexity of a range space is the VC dimension.

\begin{definition}
    The VC dimension of a range space $(X,\mathcal{R})$ is the cardinality of a maximum cardinality subset $Y \subseteq X$, such that $\lvert \{ Y \cap R \mid R \in \mathcal{R} \} \rvert = 2^{\lvert Y \rvert}$.
\end{definition}

Note that $F$ and $\overline{F}$ have equal VC dimension and for any $Y \subseteq X$, the projection of $F$ onto $Y$ has VC dimension at most the VC dimension of $F$, see for example \cite{Har-Peled2011B}. We define $(\epsilon, \eta)$-approximations of range spaces.

\begin{definition}[{\cite[Definition 2.3]{Har-Peled2011}}]
    \label{def:etaepsilonapprox}
    Let $\epsilon, \eta \in (0,1)$ and $(X,\mathcal{R})$ be a range space with finite non-empty ground set. An $(\eta, \epsilon)$-approximation of $(X,\mathcal{R})$ is a set $S \subseteq X$, such that for all $R \in \mathcal{R}$
    \begin{align*}
        \left\lvert \frac{\lvert R \cap X \rvert}{\lvert X \rvert} - \frac{\lvert R \cap S \rvert}{\lvert S \rvert} \right\rvert \leq 
        \begin{cases}
            \epsilon \cdot \frac{\lvert R \cap X \rvert}{\lvert X \rvert}, & \text{if } \lvert R \cap X \rvert \geq \eta \cdot \lvert X \rvert \\
            \epsilon \cdot \eta, & \text{else}.
        \end{cases}
    \end{align*}
\end{definition}

The following theorem is useful for obtaining $(\epsilon, \eta)$-approximations.

\begin{theorem}[{\cite[Theorem 2.11]{Har-Peled2011}}]
    \label{theo:etaepsilonapprox}
    Let $(X,\mathcal{R})$ be a range space with finite non-empty ground set and VC dimension $\mathcal{D}$. Also, let $\epsilon, \delta, \eta \in (0,1)$. There is an absolute constant $c \in \mathbb{R}_{>0}$ such that a sample of $$ \frac{c}{\eta \cdot \epsilon^2} \cdot \left( \mathcal{D} \log \left(\frac{1}{\eta}\right) + \log\left(\frac{1}{\delta}\right) \right) $$ elements drawn independently and uniformly at random with replacement from $X$ is a $(\eta, \epsilon)$-approximation for $(X, \mathcal{R})$ with probability at least $1 - \delta$.
\end{theorem}

We define open metric balls, which are the ranges used to derive our result.

\begin{definition}
    \label{def:open_metric_ball}
    For $r \in \mathbb{R}_{\geq 0}$, $z \in Z$ and $Y \subseteq X$ we denote by $\pball(z,r,Y) = \{ y \in Y \mid \rho(y,z) < r \}$ the open metric ball with center $z$ and radius $r$. We denote the set of all open metric balls by $\mathbb{B}(Y, Z) = \{ \pball(z,r,Y) \mid z \in Z, r \in \mathbb{R}_{\geq 0} \}$.
\end{definition}

Now, we are ready to analyze the computation of the actual $\epsilon$-coresets. We use the reduction to uniform sampling, introduced by Feldman and Langberg \cite{Feldman2011} and improved by Braverman et al. \cite{DBLP:journals/corr/BravermanFL16} (using \cref{theo:etaepsilonapprox}). Preferably we would apply Theorem 31 by Feldman et al. \cite{DBLP:journals/siamcomp/FeldmanSS20}, which however is not possible since it depends on a range space where each function $f \in F$ may be assigned a distinct scaling factor. This is incompatible with the range space induced by the open metric balls we use to obtain our result. However, by adapting and modifying the proof of their theorem we can derive the desired and more versatile result. To handle necessary scaling factors still involved in the analysis, we incorporate results by Munteanu et al. \cite{DBLP:conf/nips/MunteanuSSW18} for bounding the VC dimension.

\begin{theorem}
    \label{theo:coreset}
    For $f \in F$ we let $\lambda(f) = \left\lceil \lvert F \rvert \cdot 2^{\lceil \log_2(\gamma(f)) \rceil} \right\rceil / \lvert F \rvert$, $\Lambda = \sum_{f \in F} \lambda(f)$, $\psi(f) = \frac{\lambda(f)}{\Lambda}$ and $\mathcal{D}$ be the VC dimension of the range space $(Y, \mathbb{B}(Y,Z))$. Let $\delta, \epsilon \in (0,1)$. A sample $S$ of $\Theta\left(\epsilon^{-2} \alpha k^\prime (\mathcal{D} k \log(k) \log(\alpha k^\prime n) \log(\alpha k^\prime) + \log(1/\delta))\right)$ elements $\tau_i$ from $T$, drawn independently with replacement with probability $\psi(f_i)$ and weighted by $w(f_i) = \frac{\Lambda}{\lvert S \rvert \lambda(f_i)}$ is an $\epsilon$-coreset with probability at least $1 - \delta$. 
\end{theorem}
\begin{proof}
    We define for $C \subseteq Z$ with $\lvert C \rvert = k$ the estimator $$\scost{S}{C} = \sum_{\tau_i \in S} w(f_i) \cdot \min_{c \in C} \rho(\tau_i, c) = \sum_{\tau_i \in S} w(f_i) \cdot f_i(C) = \sum_{\tau_i \in S} \frac{\Lambda}{\lvert S \rvert\lambda(f_i)} f_i(C)$$ for $\cost{T}{C}$. We see that $$\pexpected[\scost{S}{C}] = \sum_{i=1}^{\lvert S \rvert} \sum_{\tau_j \in T} \frac{\Lambda}{\lvert S \rvert\lambda(f_j)} f_j(C) \frac{\lambda(f_j)}{\Lambda} = \sum_{\tau_j \in T} f_j(C) = \cost{T}{C},$$ thus $\scost{S}{C}$ is unbiased. We want to bound the error of $\scost{S}{C}$ by applying \cref{theo:etaepsilonapprox}. To do so, we reduce the sensitivity sampling to uniform sampling as follows: We let $G$ be a multiset that is a copy of $F$, where each $f \in F$ is contained $\lvert F \rvert \lambda(f)$ times and is scaled by $\frac{1}{\lvert F \rvert \lambda(f)}$. Note that $\lvert G \rvert = \lvert F \rvert \Lambda$. Also, $\psi(f) = \frac{\lvert F \rvert \lambda(f)}{\lvert G \rvert}$, $\lvert F \rvert \lambda(f)$ is integral for each $f \in F$ and $$\sum_{g \in G} g(C) = \sum_{f \in F} \frac{\lvert F \rvert \lambda(f)}{\lvert F \rvert \lambda(f)} f(C) = \sum_{f \in F} f(C) = \cost{T}{C}.$$
    
    Given a sample $S^\prime$, with $\lvert S^\prime \rvert = \lvert S \rvert$, drawn independently and uniformly at random with replacement from $G$, for $C \subseteq Z$ with $\lvert C \rvert = k$ we define the estimator $$\ucost{S^\prime}{C} = \frac{\lvert G \rvert}{\lvert S^\prime \rvert} \sum_{g \in S^\prime} g(C)$$ for $\cost{T}{C}$. We see that
    \begin{align*}
        \pexpected\left[\ucost{S^\prime}{C}\right] & = \frac{\lvert G \rvert}{\lvert S^\prime \rvert} \sum_{i=1}^{\lvert S^\prime \rvert} \sum_{f \in F} \frac{f(C)}{\lvert F \rvert \lambda(f)} \frac{\lvert F \rvert \lambda(f)}{\lvert G \rvert} = \frac{1}{\lvert S^\prime \rvert} \sum_{i=1}^{\lvert S^\prime \rvert} \sum_{f \in F} f(C) = \sum_{g \in G} g(C).
    \end{align*}
    Thus, $\ucost{S^\prime}{C}$ is unbiased, too. We now assume that $S^\prime = \left\{ \frac{1}{\lvert F \rvert\lambda(f_i)} \cdot f_i \mid \tau_i \in S \right\}$. Then, 
    \begin{align*}
        \ucost{S^\prime}{C} & = \frac{\lvert G \rvert}{\lvert S^\prime \rvert} \sum_{g \in S^\prime} g(C) = \frac{\lvert F \rvert \Lambda}{\lvert S^\prime \rvert} \sum_{\tau_i \in S} \frac{1}{\lvert F \rvert \lambda(f_i)} f_i(C) = \sum_{\tau_i \in S} \frac{\Lambda}{\lvert S \rvert \lambda(f_i)} f_i(C) \\
        & = \scost{S}{C},
    \end{align*}
    so the error bound for $\ucost{S^\prime}{C}$, that we derive in the following, also applies to $\scost{S}{C}$, hence $S$ together with $w$ is a weighted $\epsilon$-coreset (see \cref{def:coreset}). We now apply \cref{theo:etaepsilonapprox} with the given $\delta$, $\epsilon/2$ and $\eta = 1/\Lambda$, so the overall error is at most $\epsilon \cdot \cost{T}{C}$ for each $C \subseteq Z$ with $\lvert C \rvert = k$.
    
    For $H \subseteq G$, $C \subseteq Z$ and $r \in \mathbb{R}_{\geq 0}$, we let $\range{H}{C}{r} = \{ g \in H \mid g(C) \geq r \}$. Now, we let $(G, \mathcal{R})$ be a range space over $G$, where $\mathcal{R} = \{ \range{G}{C}{r} \mid r \in \mathbb{R}_{\geq 0}, C \subseteq Z, \lvert C \rvert = k \}$. For all $C \subseteq Z$ with $\lvert C \rvert = k$ and all $H \subseteq G$ we have that
    \begin{align*}
        \sum_{g \in H} g(C) & = \sum_{g \in H} \int_0^\infty \ind{g(C) \geq r}\ \mathrm{d}r = \int_0^\infty \sum_{g \in H} \ind{g(C) \geq r} \ \mathrm{d}r \\
        & = \int_0^\infty \lvert \range{H}{C}{r} \rvert \ \mathrm{d}r. \tag{I}\label{eq:intidentity}
    \end{align*}
    Note that the indicator function is integrable under this circumstances and $\lvert \range{H}{C}{r} \rvert$ is a step function and is integrable, too. Using this identity, for all $C \subseteq Z$ with $\lvert C \rvert = k$ we now bound the error introduced by $\scost{S}{C}$:
    \begin{align*}
        \left\lvert \cost{T}{C} - \scost{S}{C} \right\rvert & = \left\lvert \cost{T}{C} - \ucost{S^\prime}{C} \right\rvert = \left\lvert \sum_{g \in G} g(C) - \frac{\lvert G \rvert}{\lvert S^\prime \rvert} \sum_{g \in S^\prime} g(C) \right\rvert \\ 
        & = \left\lvert \int_0^\infty \lvert \range{G}{C}{r} \rvert \ \mathrm{d} r -  \frac{\lvert G \rvert}{\lvert S^\prime \rvert} \int_0^\infty \lvert \range{S^\prime}{C}{r} \rvert \ \mathrm{d} r \right\rvert \\
        & = \left\lvert \int_0^\infty \lvert \range{G}{C}{r} \rvert - \frac{\lvert G \rvert}{\lvert S^\prime \rvert} \lvert \range{S^\prime}{C}{r} \rvert \ \mathrm{d}r \right \rvert \\
        & \leq \int_0^\infty \left\lvert \lvert \range{G}{C}{r} \rvert - \frac{\lvert G \rvert}{\lvert S^\prime \rvert} \lvert \range{S^\prime}{C}{r} \rvert \right \rvert \mathrm{d}r.
    \end{align*}
    Here the second equation follows from \cref{eq:intidentity}.
    
    In the following, let $\serror{C}{r} = \left\lvert \lvert \range{G}{C}{r} \rvert - \frac{\lvert G \rvert}{\lvert S^\prime \rvert} \lvert \range{S^\prime}{C}{r} \rvert \right \rvert$, $r_u(C) = \max\limits_{g \in G} g(C)$, $R_1(C) = \{ r \in \mathbb{R}_{\geq 0} \mid \lvert \range{G}{C}{r} \rvert \geq \eta \cdot \lvert G \rvert \}$ and $R_2(C) = \mathbb{R}_{\geq 0} \setminus R_1(C)$. Note that $R_1(C)$ and $R_2(C)$ are intervals due to the monotonicity of $\lvert \range{G}{C}{r} \rvert$. Furthermore, for $r \in (r_u(C), \infty)$ it holds that $\lvert \range{G}{C}{r} \rvert = 0$. Using these facts, we further derive:
    \begin{align*}
        \int_0^\infty \serror{C}{r} \ \mathrm{d}r & = \int_{R_1} \serror{C}{r} \ \mathrm{d}r + \int_{R_2} \serror{C}{r} \ \mathrm{d}r \\
        & \leq \int_{R_1} \frac{\epsilon}{2} \lvert \range{G}{C}{r} \rvert \ \mathrm{d}r + \int_{R_2} \frac{\epsilon}{2} \eta \lvert G \rvert \ \mathrm{d}r \\
        & \leq \frac{\epsilon}{2} \int\limits_{0}^{\infty} \lvert \range{G}{C}{r} \rvert \ \mathrm{d}r + \frac{\epsilon \eta \lvert G \rvert}{2} \int\limits_{0}^{r_u(C)} \ \mathrm{d} r \\
        & = \frac{\epsilon}{2} \sum_{g \in G} g(C) + \frac{\epsilon \eta \lvert G \rvert r_u(C)}{2}. \tag{III} \label{eq:half_bound}
    \end{align*}
    Here, the first inequality follows from \cref{def:etaepsilonapprox} and in the last equation we use \cref{eq:intidentity}. Finally, we bound the last summand in \cref{eq:half_bound}. First note that we have for each $g \in G$:
    \begin{align*}
        \frac{g(C)}{\sum_{h \in G} h(C)} = \frac{\frac{1}{\lvert F \rvert \lambda(f)} f(C)}{\sum_{h \in F} h(C)} \leq \frac{\lambda(f)}{\lvert F \rvert \lambda(f)} & \iff \frac{g(C)}{\sum_{h \in G} h(C)} \leq \frac{1}{\lvert F \rvert},
    \end{align*}
    where $f \in F$ is the function that $g$ is a copy of and the inequality follows from \cref{def:sensitivity}. This implies $r_u(C) \leq \frac{1}{\lvert F \rvert} \sum_{h \in G} h(C)$. We now further derive:
    \begin{align*}
        \frac{\epsilon \eta \lvert G \rvert r_u(C)}{2} & \leq \frac{\epsilon}{2} \frac{1}{\Lambda} \lvert F \rvert \Lambda \frac{1}{\lvert F \rvert} \sum_{g \in G} g(C) = \frac{\epsilon}{2} \sum_{g \in G} g(C).
    \end{align*}
    So, all in all $\lvert \cost{T}{C} - \scost{S}{C} \rvert \leq \epsilon \cdot \cost{T}{C}$ for all $C \subseteq Z$ with $\lvert C \rvert = k$.
    \\
    
    The claim now follows from the facts that 
    \begin{itemize}
        \item $\gamma(f) \leq \lambda(f) \leq 2 \cdot \gamma(f) + \frac{1}{\lvert F \rvert}$ for each $f \in F$, thus $\Lambda \leq 2 \cdot \Gamma(F) + 1$ and $\Gamma(F) \in O(\alpha k^\prime)$ by \cref{lem:sensitivities} and
        \item $(G, \mathcal{R})$ has VC dimension in $O(\mathcal{D} k \log(k) \log(\alpha k^\prime n))$ by the following \cref{lem:vc_g}.
    \end{itemize}
\qed\end{proof}

\begin{lemma}
    \label{lem:vc_g}
    $(G, \mathcal{R})$ has VC dimension $O(\mathcal{D} k \log(k) \log(\alpha k^\prime n))$, where $\mathcal{D}$ is the VC dimension of the range space $(Y, \mathbb{B}(Y,Z))$.
\end{lemma}
\begin{proof}
    First, we assume that there exists a $\varphi \in \mathbb{R}$, such that $\lambda(f) = \varphi$ for all $f \in F$. Then all functions in $G$ are scaled uniformly and we can completely neglect the scaling. In this case $(G, \mathcal{R})$ has equal VC dimension as $$Q_1 = (T, \{ T \setminus (\pball(c_1, r, T) \cup \dots \cup \pball(c_k, r, T)) \mid \{c_1, \dots, c_k\} \subseteq Z, r \in \mathbb{R}_{\geq 0} \}),$$ the VC dimension of $Q_1$ is at most the VC dimension of $$Q_2 = (T, \{T \setminus (B_1 \cup \dots \cup B_k) \mid B_1, \dots, B_k \in \mathbb{B}(T,Z) \}),$$ $Q_2$ is the projection of $$Q_3 = (Y, \{ Y \setminus (B_1 \cup \dots \cup B_k) \mid B_1, \dots, B_k \in \mathbb{B}(Y, Z)\})$$ onto $T$ and has thus at most the VC dimension of $Q_3$ and finally, $Q_3$ is the complementary range space of $$Q_4 = (Y, \{ B_1 \cup \dots \cup B_k \mid B_1, \dots, B_k \in \mathbb{B}(Y, Z)\})$$ and has thus equal VC dimension as $Q_4$. By the $k$-fold union theorem \cite[Lemma 2.3.2]{DBLP:journals/jacm/BlumerEHW89} $Q_4$ has VC dimension $O(\mathcal{D} \cdot k \log(k))$, where $\mathcal{D}$ is the VC dimension of $(Y, \mathbb{B}(Y,Z))$. For the following, let $c$ be the constant hidden in this O-notation.
    \\
    
    Contrary to the former case, if there are $t > 1$ distinct values $\Phi = \{ \varphi_1, \dots, \varphi_t \} \subset \mathbb{R}$, such that $\lambda(f) \in \Phi$ for each $f \in F$ and $\forall i \in [t] \exists f \in F: \lambda(f) = \varphi_i$, we apply the techniques of Munteanu et al. \cite{DBLP:conf/nips/MunteanuSSW18} (see Lemma 11 and Theorem 15 therein).
    \\
    
    First, assume that the VC dimension of $(G, \mathcal{R})$ is greater than $t \cdot c \cdot \mathcal{D} \cdot k \log(k)$. Hence there exists a set $G^\prime \subseteq G$ with $\lvert G^\prime \rvert > t \cdot c \cdot \mathcal{D} \cdot k \log(k)$, such that $\lvert \{ G^\prime \cap R \mid R \in \mathcal{R} \} \rvert = 2^{\lvert G^\prime \rvert}$. Let $\{ G_1, \dots, G_t \}$ be a partition of $G$, such that for each $g \in G_i$ we have $g = \frac{1}{\lvert F \rvert \lambda(f)} f = \frac{1}{\lvert F \rvert \varphi_i} f$ for an $f \in F$. Furthermore, for $i \in [t]$, let $G^\prime_i = G^\prime \cap G_i$. 
    
    By disjointness, we have $\lvert \{ G^\prime_i \cap R_i \mid R_i = (R \cap G_i), R \in \mathcal{R} \} \rvert = 2^{\lvert G^\prime_i \rvert}$ for each $i \in [t]$ and also there must exist at least one $j \in [t]$, such that $\lvert G^\prime_j \rvert \geq \frac{\lvert G^\prime \rvert}{t} > \frac{t \cdot c \cdot \mathcal{D} \cdot k \log(k)}{t} = c \cdot \mathcal{D} \cdot k \log(k)$, hence the projection of $(G, \mathcal{R})$ on $G_j$ has VC dimension greater than $c \cdot \mathcal{D} \cdot k \log(k)$. This is a contradiction to the former case of uniformly scaled functions in $G$, thus $(G, \mathcal{R})$ has VC dimension $O(t \cdot \mathcal{D} \cdot k \log(k))$ in this case.
    \\
    
    Now, we bound $t$. Recall that for each $f \in F$, $\lambda(f) = \left\lceil \lvert F \rvert 2^{\lceil \log_2(\gamma(f)) \rceil} \right\rceil/\lvert F \rvert$. Furthermore for each $i \in [k^\prime]$ and $\tau_j \in \hat{V}_i$ (see \cref{lem:sensitivities}), 
    \begin{align*}
        \left(1 + \sqrt{\frac{3\alpha}{2k^\prime}} \right) \frac{2}{\lvert \hat{V}_i \rvert} \leq \gamma(f_j) \leq \left(1 + \sqrt{\frac{2k^\prime}{3\alpha}} \right) \left(\alpha + \frac{2\alpha}{\lvert \hat{V}_i \rvert} \right) + \left(1 + \sqrt{\frac{3\alpha}{2k^\prime}} \right) \frac{2}{\lvert \hat{V}_i \rvert}.
    \end{align*}
    Therefore, there can be at most 
    \begin{align*}
        & \log_2\left(\frac{\left(1 + \sqrt{\frac{2k^\prime}{3\alpha}} \right) \left(\alpha + \frac{2\alpha}{\lvert \hat{V}_i \rvert} \right) + \left(1 + \sqrt{\frac{3\alpha}{2k^\prime}} \right) \frac{2}{\lvert \hat{V}_i \rvert}}{\left(1 + \sqrt{\frac{3\alpha}{2k^\prime}} \right) \frac{2}{\lvert \hat{V}_i \rvert}}\right) \\
        \leq & \log_2\left(\frac{\left(1 + \sqrt{\frac{2k^\prime}{3\alpha}} \right)}{\left(1 + \sqrt{\frac{3\alpha}{2k^\prime}} \right)} (\alpha n/2 + \alpha) + 1\right) \\
        \leq & \log_2(9\alpha k^\prime(\alpha n/2 + \alpha) + 1)
    \end{align*}
    distinct values of $2^{\lceil \log_2(\gamma(f))\rceil}$, which upper bounds the number of distinct values of $\lambda(f)$. 
    
    We conclude that the VC dimension of $(G, \mathcal{R})$ is in $O(\mathcal{D} k \log(k) \log(\alpha k^\prime n))$.
\qed\end{proof}

\section{Coresets for $(k,\ell)$-Median Clustering under the Fréchet Distance}
\label{sec:coresets_frechet}

Now we present an algorithm for computing $\epsilon$-coresets for $(k, \ell)$-median clustering of polygonal curves under the Fr\'echet distance. We start by defining polygonal curves.

\begin{definition}
    \label{def:polygonal_curve}
	A (parameterized) curve is a continuous mapping $\tau \colon [0,1] \rightarrow \mathbb{R}^d$. A curve $\tau$ is polygonal, iff there exist $v_1, \dots, v_m \in \mathbb{R}^d$, no three consecutive on a line, called $\tau$'s vertices, and $t_1, \dots, t_m \in [0,1]$ with $t_1 < \dots < t_m$, $t_1 = 0$ and $t_m = 1$, called $\tau$'s instants, such that $\tau$ connects every two consecutive vertices $v_i = \tau(t_i), v_{i+1} = \tau(t_{i+1})$ by a line segment.
\end{definition}
We call the segments $\overline{v_1v_2}, \dots, \overline{v_{m-1}v_m}$ edges of $\tau$ and $m$ the complexity of $\tau$, denoted by $\lvert \tau \rvert$. 

\begin{definition}
    \label{def:frechet_distance}
    Let $\mathcal{H}$ denote the set of all continuous bijections $h\colon [0,1] \rightarrow [0,1]$ with $h(0) = 0$ and $h(1) = 1$, which we call reparameterizations.
    The Fr\'echet distance between curves $\sigma$ and $\tau$ is $ d_F(\sigma, \tau) = \inf_{h \in \mathcal{H}}\  \max_{t \in [0,1]}\ \lVert \sigma(t) - \tau(h(t)) \rVert$.
\end{definition}

Now we introduce the classes of curves we are interested in.

\begin{definition}
    For $d \in \mathbb{N}$, we define by $\mathbb{X}^d$ the set of equivalence classes of polygonal curves (where two curves are equivalent, iff they can be made identical by a reparameterization) in ambient space $\mathbb{R}^d$. For $m \in \mathbb{N}$ we define by $\mathbb{X}^d_m$ the subclass of polygonal curves of complexity at most $m$.
\end{definition}

Finally, we define the $(k,\ell)$-median clustering problem for polygonal curves.

\begin{definition}
    The $(k,\ell)$-median clustering problem is defined as follows, where $k,\ell \in \mathbb{N}$ are fixed (constant) parameters of the problem: given a set $T \subset \mathbb{X}^d_m$ of $n$ polygonal curves, compute a set of $k$ curves $C^\ast \subset \mathbb{X}^d_\ell$, such that $\cost{T}{C^\ast} = \sum\limits_{\tau \in T} \min\limits_{c^\ast \in C^\ast} d_F(\tau, c^\ast)$ is minimal.
\end{definition}

We bound the VC dimension of metric balls under the Fr\'echet distance by showing that a result of Driemel et al. \cite{Driemel19} holds also in our setting.

\begin{theorem}
    \label{theo:open_frechet_ball_vc}
    The VC dimension of $(\mathbb{X}^d_m, \mathbb{B}(X^d_m, X^d_\ell))$ is $O\left(\ell^2 \log(\ell m) \right)$.
\end{theorem}
\begin{proof}
    We argue that the claim follows from Theorem 18 by Driemel et al. \cite{Driemel19}. First, in their paper polygonal curves do not need to adhere the restriction that no three consecutive vertices may be collinear and they define $\mathbb{X}^d_m$ to be the polygonal curves of exactly $m$ vertices. However, our definitions match by simulating the addition of collinear vertices to those curves in $\mathbb{X}^d_m$ with less than $m$ vertices. 
    
    Now, looking into their proof, we can slightly modify the geometric primitives by letting $B_r(p) = \{ x \in \mathbb{R}^d \mid \lVert x - p \rVert < r \}$, $D_r(\overline{st}) = \{x \in \mathbb{R}^d \mid \exists p \in \overline{st}: \lVert p - x \rVert < r\}$, $C_r(\overline{st}) = \{x \in \mathbb{R}^d \mid \exists p \in \ell(\overline{st}): \lVert p - x \rVert < r\}$ and $R_r(\overline{st}) = \{ p + u \mid p \in \overline{st}, u \in \mathbb{R}^d, \langle t - s, u \rangle = 0, \lVert u \rVert < r\}$, which does not affect the remainder of the proof and thus yields the same bound on the VC dimension.
\qed\end{proof}

To compute $\epsilon$-coresets for $(k,\ell)$-median clustering under the Fr\'echet distance, we first need to compute the sensitivities and to do so, we utilize constant factor approximation algorithms. We use \cite[Algorithm 1]{doi:10.1137/1.9781611976465.160}, which only works for $k=1$ but is very efficient in this case. For $k > 1$ we use \cref{algo:capprox}, a modification of \cite[Algorithm 3]{DBLP:conf/soda/DriemelKS16}, which we now present. This algorithm uses (approximate) minimum-error $\ell$-simplifications, which we now define.

\begin{definition}
    An $\alpha$-approximate minimum-error $\ell$-simplification of a polygonal curve $\tau \in \mathbb{X}^d$ is a curve $\sigma \in \mathbb{X}^d_\ell$ with $d_F(\tau, \sigma) \leq \alpha \cdot d_F(\tau, \sigma^\prime)$ for all $\sigma^\prime \in \mathbb{X}^d_\ell$.
\end{definition}

The following lemma is useful to obtain simplifications.

\begin{lemma}[{\cite[Lemma 7.1]{k_l_center}}]
    \label{lem:imai_iri_simpli}
    Given a curve $\sigma \in \mathbb{X}^d_m$, a $4$-approximate minimum-error $\ell$-simplification can be computed in $O(m^3 \log m)$ time, by combining the algorithms by Alt and Godau \cite{alt_godau} and Imai and Iri \cite{imai_polygonal_1988}.
\end{lemma}

We now present the constant factor approximation algorithm.

\begin{algorithm}[H]
	\caption{Constant Factor Approximation for \((k,\ell)\)-Median Clustering}\label{algo:capprox}
	\begin{algorithmic}[1]
		\Procedure{\((k,\ell)\)-Median-96-Approximation}{$T = \{\tau_1, \dots, \tau_n\}$}
			\For{\(i=1, \dots, n\)}
				\State \(\hat{\tau}_i \gets\) approximate minimum-error \(\ell\)-simplification of \(\tau_i\)
			\EndFor
			\State \(C \gets\) Chen's algorithm with \(\epsilon = 0.5, \lambda = \delta\) on \(\{\hat{\tau}_1, \dots, \hat{\tau}_n \}\) \cite[Theorem 6.2]{chen_coresets_2009}
			\State \Return \(C\)
		\EndProcedure
	\end{algorithmic}
\end{algorithm}

We prove the correctness and analyze the running time of \cref{algo:capprox}.

\begin{theorem}
    \label{theo:capprox_frechet}
    Given $\delta \in (0,1)$ and $T = \{ \tau_1, \dots, \tau_n \} \subset \mathbb{X}^d_m$, \cref{algo:capprox} returns with probability at least $1-\delta$ a $109$-approximate $(k,\ell)$-median solution for $T$ in time $O(n m \log(1/\delta)\log(m) + n m^3 \log(m))$.
\end{theorem}
\begin{proof}
    We assume that the approximate minimum-error \par $\ell$-simplifications are computed combining the algorithms by Alt and Godau \cite{alt_godau} and Imai and Iri \cite{imai_polygonal_1988}, so the approximation factor is $4$ by \cref{lem:imai_iri_simpli}. 
    
    For $i \in [n]$, let $\hat{\tau}_i$ be the simplification of $\tau_i$ and let $\hat{T} = \{ \hat{\tau}_1, \dots, \hat{\tau}_n \}$. Note that $d_F(\tau_i,\hat{\tau}_i) \leq 4 \cdot d_F(\tau, \sigma)$ for all $i \in [n]$ and $\sigma \in \mathbb{X}^d_\ell$.
    
    Let $\hat{C}^\ast = \{ \hat{c}^\ast_1, \dots, \hat{c}^\ast_k \} \subset \mathbb{X}^d_\ell$ be an optimal $(k,\ell)$-median solution for $\hat{T}$ and let $C^\prime = \{c^\prime_1, \dots, c^\prime_k\} \subseteq \hat{T}$ be an optimal solution to the discrete $(k,\ell)$-median problem for $\hat{T}$, i.e. the centers are chosen among the input. Breaking ties arbitrarily, we assume that every $\hat{\tau} \in \hat{T}$ has a unique nearest neighbor in $\hat{C}^\ast$ and for $i \in [k]$ we define $\hat{T}_i = \{ \hat{\tau} \in \hat{T} \mid \forall j \in [k] : d_F(\hat{\tau}, \hat{c}^\ast_i) \leq d_F(\hat{\tau}, \hat{c}^\ast_j) \}$, such that $\hat{T}_1, \dots, \hat{T}_k$ form a partition of $\hat{T}$. By the triangle inequality:
    \begin{align*}
        \cost{\hat{T}}{C^\prime} & = \min_{\substack{C \subseteq \hat{T} \\ \lvert C \rvert = k }} \sum_{i=1}^k \sum_{\hat{\tau} \in \hat{T}_i} \min_{c \in C} d_F(\hat{\tau}, c) \leq \min_{\substack{C \subseteq \hat{T} \\ \lvert C \rvert = k }} \sum_{i=1}^k \sum_{\hat{\tau} \in \hat{T}_i} (d_F(\hat{\tau}, \hat{c}^\ast_i) + \min_{c \in C} d_F(\hat{c}^\ast_i, c)) \\
        & = \cost{\hat{T}}{\hat{C}^\ast} + \min_{\substack{C \subseteq \hat{T} \\ \lvert C \rvert = k }} \sum_{i=1}^k \sum_{\hat{\tau} \in \hat{T}_i}  \min_{c \in C} d_F(\hat{c}^\ast_i, c) \\
        & = \cost{\hat{T}}{\hat{C}^\ast} + \sum_{i=1}^k \sum_{\hat{\tau} \in \hat{T}_i} \min_{\hat{\sigma} \in \hat{T}_i} d_F(\hat{c}^\ast_i, \hat{\sigma}).
    \end{align*}
    For each $i \in [k]$ there must exist a $\hat{\sigma} \in \hat{T}_i$ with $d_F(\hat{\sigma}, \hat{c}^\ast_i) \leq \sum_{\hat{\tau} \in \hat{T}_i} d_F(\hat{\tau}, \hat{c}^\ast_i)/\lvert \hat{T}_i \rvert$, since otherwise $\sum_{i=1}^k \sum_{\hat{\tau} \in \hat{T}_i} \min_{\hat{\sigma} \in \hat{T}_i} d_F(\hat{c}^\ast_i, \hat{\sigma}) > \cost{\hat{T}}{\hat{C}^\ast}$, which is a contradiction. We conclude that $\cost{\hat{T}}{C^\prime} \leq 2 \cost{\hat{T}}{\hat{C}^\ast}$. Also, by \cite[Theorem 6.2]{chen_coresets_2009} $\cost{\hat{T}}{C} \leq 10.5 \cost{\hat{T}}{C^\prime}$. 
    
    Now, let $C^\ast = \{c^\ast_1, \dots, c^\ast_k \} \subset \mathbb{X}^d_\ell$ be an optimal $(k,\ell)$-median solution for $T$ and $C = \{c_1, \dots, c_k\}$ be a solution returned by \cref{algo:capprox} for $\hat{T}$. We derive:
    \begin{align*}
        \cost{T}{C} & \leq \sum_{i=1}^n (d_F(\tau_i, \hat{\tau}_i) + \min_{j \in [k]} d_F(\hat{\tau}_i, c_j)) = \sum_{i=1}^n d_F(\tau_i, \hat{\tau}_i) + \cost{\hat{T}}{C} \\
        & \leq 4 \cost{T}{C^\ast} + 21 \cost{\hat{T}}{\hat{C}^\ast} \leq 4 \cost{T}{C^\ast} + 21 \cost{\hat{T}}{C^\ast} \\
        & \leq 4 \cost{T}{C^\ast} + 21 \left( \sum_{i=1}^n (d_F(\hat{\tau}_i, \tau_i) + \min_{j \in [k]} d_F(\tau_i, c^\ast_j)) \right) \\
        & \leq 4 \cost{T}{C^\ast} + 84 \cost{T}{C^\ast} + 21 \cost{T}{C^\ast} \leq 109 \cost{T}{C^\ast}.
    \end{align*}

    Computing the simplifications takes time $O(n m^3 \log(m))$, see \cref{lem:imai_iri_simpli}. Further, we incorporate the given probability of failure (see \cite[Theorem 3.6]{chen_coresets_2009}) into the running time stated in \cite[Theorem 6.2]{chen_coresets_2009}. Hence, Chen's algorithm can be run in time $O(n m \log(1/\delta) \log(m))$ when the distances are computed using Alt and Godau's algorithm \cite{alt_godau}.
\qed\end{proof}

We now present the algorithm for computing weighted $\epsilon$-coresets for $(k,\ell)$-median clustering.

\begin{algorithm}[H]
	\caption{Coresets for \((k,\ell)\)-Median Clustering\label{algo:coreset}}
	\begin{algorithmic}[1]
		\Procedure{\((k, \ell)\)-Median-Coreset}{$T = \{\tau_1, \dots, \tau_n\}, \delta, \epsilon$}
			\If{\(k=1\)}
				\State \(\hat{c} \gets\) \(\ell\)-Median-\(34\)-Approximation\((T,\delta/2)\) \cite[Algorithm 1]{doi:10.1137/1.9781611976465.160}
				\State \(\hat{C} = \{ \hat{c} \}\)
			\Else
				\State \(\hat{C} = \{\hat{c}_1, \dots, \hat{c}_k\} \gets \) \cref{algo:capprox}\((T, \delta/2)\)
			\EndIf
			\State compute \(\hat{V}_1, \dots, \hat{V}_k\), \(\hat{\Delta}_1, \dots, \hat{\Delta}_k\) and \(\gamma\) w.r.t. \(\hat{C}\) (\confer \cref{lem:sensitivities})
			\State compute \(\lambda\), \(\Lambda\) w.r.t. \(\gamma\) and \(\psi\) w.r.t. \(\lambda\) (\confer \cref{theo:coreset})
			\State \(S \gets\) sample \(\Theta(k\epsilon^{-2}(d^2 \ell^2 k \log(d \ell m) \log(kn) \log^2(k) + \log(1/(2\delta))))\) 
			
			\hspace{\algorithmicindent} elements from \(T\) independently with replacement with respect to \(\psi\)
			\State compute \(w\) w.r.t. \(\lambda\), \(\Lambda\) and \(S\) (\confer \cref{theo:coreset})
			\State \Return \(S\) and \(w\)
		\EndProcedure
	\end{algorithmic}
\end{algorithm}

We prove the correctness and analyze the running time of \cref{algo:coreset}. Also, we analyze the size of the resulting $\epsilon$-coreset.

\begin{theorem}
    \label{theo:kl_coreset}
    Given a set $T = \{ \tau_1, \dots, \tau_n \} \subset \mathbb{X}^d_m$ and $\delta, \epsilon \in (0,1)$, \cref{algo:coreset} computes a weighted $\epsilon$-coreset of size $O(\epsilon^{-2}(\log(m) \log(n) + \log(1/\delta)))$ for $(k, \ell)$-median clustering with probability at least $1-\delta$, in time $$O(n m \log(m) \log(1/\delta) + n m^3 \log (m) + \epsilon^{-2}(\log(m)\log(n) + \log(1/\delta)))$$ for $k > 1$ and $$O(n m \log(m) + m^2 \log(m) \log^2(1/\delta) + m^3 \log(m) + \epsilon^{-2}(\log(m)\log(n) + \log(1/\delta)))$$ for $k = 1$.
\end{theorem}
\begin{proof}
    First note that by a union bound, with probability at least $1-\delta$ \cref{algo:capprox}, respectively \cite[Algorithm 1]{doi:10.1137/1.9781611976465.160}, and the sampling are simultaneously successful (see \cref{theo:capprox_frechet,theo:coreset} and \cite[Corollary 3.1]{doi:10.1137/1.9781611976465.160}). The correctness of the algorithm follows from the observations that the $(k,\ell)$-median clustering objective fits the generalized $k$-median clustering objective with $X = \mathbb{X}^d$, $Y = \mathbb{X}^d_m \subset X$ and $Z = \mathbb{X}^d_\ell \subset X$, therefore \cref{lem:sensitivities} and \cref{theo:coreset} can be applied, and the VC dimension of $(X, \mathbb{B}(Y,Z))$ is $O(\ell^2 \log(\ell m))$ by \cref{theo:open_frechet_ball_vc}.
    
    We now analyze the running time. $\hat{V}_1, \dots, \hat{V}_k$, $\hat{\Delta}_1, \dots, \hat{\Delta}_k$ and $\gamma$ can be computed in time $O(n m \log(m))$ using Alt and Godau's algorithm \cite{alt_godau}. $\lambda$, $\Lambda$ and $\psi$ can be computed in time $O(n)$ and the sampling can be carried out in time $O(\epsilon^{-2}(\log(m)\log(n) + \log(1/\delta)))$. Finally, $w$ can be computed in time $O(n)$. If $k > 1$ we run \cref{algo:capprox} in time $O(n m \log(1/\delta) \log(m) + n m^3 \log(m))$, see \cref{theo:capprox_frechet}. Else, we run \cite[Algorithm 1]{doi:10.1137/1.9781611976465.160} in time $(m^2 \log(m) \log^2(1/\delta) + m^3 \log(m))$, see \cite[Corollary 3.1]{doi:10.1137/1.9781611976465.160}.
    All in all, the running time is then $$O(n m \log(1/\delta) \log(m) + n m^3 \log (m) + \epsilon^{-2}(\log(m)\log(n) + \log(1/\delta)))$$ for $k > 1$ and $$O(n m \log(m) + m^2 \log(m) \log^2(1/\delta) + m^3 \log(m) + \epsilon^{-2}(\log(m) + \log(1/\delta)))$$ for $k=1$.
\qed\end{proof}

\section{Towards Practical $(1,\ell)$-Median Approximation Algorithms}
\label{sec:approximation_improvement}

In this section, we present a modification of Algorithm 3 from \cite{arxiv_full}. Our modification uses $\epsilon$-coresets to improve the running time of the algorithm, rendering it more tractable in a big data setting. We start by giving some definitions. For $p \in \mathbb{R}^d$ and $r \in \mathbb{R}_{\geq 0}$ we denote by $B(p,r) = \{ q \in \mathbb{R}^d \mid \lVert p - q \rVert \leq r\}$ the closed Euclidean ball of radius $r$ with center $p$. We give a standard definition of grids.

\begin{definition}[grid]
    \label{def:grid}
    Given a number $r \in \mathbb{R}_{>0}$, for $(p_1, \dots, p_d) \in \mathbb{R}^d$ we define by $G(p,r) = (\lfloor p_1 / r \rfloor \cdot r, \dots, \lfloor p_d / r \rfloor \cdot r)$ the $r$-grid-point of $p$. Let $P \subseteq \mathbb{R}^d$ be a subset of $\mathbb{R}^d$. The grid of cell width $r$ that covers $P$ is the set $\mathbb{G}(P,r) = \{G(p,r) \mid p \in P\}$.
\end{definition}
Such a grid partitions the set $P$ into cubic regions and for each $r \in \mathbb{R}_{>0}$ and $p \in P$ we have that $\lVert p - G(p,r) \rVert \leq \sqrt{d}r$. The following theorem by Indyk~\cite{Indyk00} is useful for evaluating the cost of a curve at hand.

\begin{theorem}[{\cite[Theorem 31]{Indyk00}}]
	\label{theo:indyk_median}
	Let $\epsilon \in (0,1]$ and $T \subset \mathbb{X}^d$ be a set of polygonal curves. Further let $W$ be a non-empty sample, drawn uniformly and independently at random from $T$, with replacement. For $\tau, \sigma \in T$ with $\cost{T}{\tau} > (1+\epsilon) \cost{T}{\sigma}$ it holds that $\Pr[\cost{W}{\tau} \leq \cost{W}{\sigma}] < \exp\left( - {\epsilon^2 \lvert W \rvert}/{64} \right)$.
\end{theorem}

The following theorem, which we combine with fine-tuned grids, allows us to obtain low-complexity center curves.

\begin{lemma}[{\cite[Lemma 4.1]{arxiv_full}}]
    \label{lem:shortcutting_single}
    Let $\sigma, \tau \in \mathbb{X}^d$ be polygonal curves. Let $v^\tau_1, \dots, v^\tau_{\lvert \tau \rvert}$ be the vertices of $\tau$ and let $r = d_F(\sigma, \tau)$. There exists a polygonal curve $\sigma^\prime \in \mathbb{X}^d$ with every vertex contained in at least one of $B(v^\tau_1, r), \dots, B(v^\tau_{\lvert \tau \rvert}, r)$, $d_F(\sigma^\prime, \tau) \leq d_F(\sigma, \tau)$ and $\lvert \sigma^\prime \rvert \leq 2 \lvert \sigma \rvert - 2$.
\end{lemma}

Finally, we present our improved modification of Algorithm 3 from \cite{arxiv_full}. This algorithm uses $\epsilon$-coresets every time it has to evaluate the cost of a center set. The dramatic effect of this small modification is that we nearly lose the original linear running time dependency on $n$ in the most time consuming part of the algorithm, rendering it practical in the setting where we have a lot of curves of much smaller complexity than number ($\ell < m \ll n$).

\begin{algorithm}[H]
	\caption{\((1,\ell)\)-Median by Simple Shortcutting and \(\epsilon\)-Coreset}\label{algo:approx}
	\begin{algorithmic}[1]
		\Procedure{\((1,\ell)\)-Median-\((5+\epsilon)\)-Approximation}{$T = \{\tau_1, \dots, \tau_n\}, \delta, \epsilon$}
			\State 	\(\hat{c} \gets\) \((1,\ell)\)-Median-\(34\)-Approximation\((T,\delta/4)\) \cite[Algorithm 1]{doi:10.1137/1.9781611976465.160}
			\State \(\epsilon^\prime \gets \epsilon/67\), \(P \gets \emptyset\)
			\State \((T^\prime, w) \gets (1,2\ell-2)\)-Median-Coreset\((T, \delta/4, \epsilon^\prime)\)
			\State \(\Delta \gets \pcost_w(T^\prime, \{\hat{c}\})\), \(\Delta_u \gets \Delta/(1-\epsilon^\prime)\), \(\Delta_l \gets \Delta / ((1+\epsilon^\prime)34)\)
			\State \(S \gets\) sample \(\left\lceil -2(\epsilon^\prime)^{-1} (\ln(\delta)-\ln(4)) \right\rceil\) curves from \(T\) uniformly 
			
			\hspace{\algorithmicindent} and independently with replacement
			\State \(W \gets\) sample \(\lceil -64(\epsilon^\prime)^{-2}(\ln(\delta) - \ln(\lceil -8 (\epsilon^\prime)^{-1} (\ln(\delta)-\ln(4))\rceil)) \rceil\) curves 
			
			\hspace{\algorithmicindent} from \(T\) uniformly and independently with replacement
			\State \(c \gets\) arbitrary element from \(\argmin_{s \in S} \cost{W}{s}\)
			\For{ \(i = 1, \dots, \lvert c \rvert\)}
				\State \(P \gets P \cup \mathbb{G}(B\left(v^c_i, (3+4\epsilon^\prime) \Delta_u/n \right), \epsilon^\prime \Delta_l/(n\sqrt{d}))\) (\(v^c_i\): \(i\)\textsuperscript{th} vertex of \(c\))
			\EndFor
			\State \(C \gets\) set of all polygonal curves with \(2\ell-2\) vertices from \(P\)
			\State \Return  \(\argmin_{c^\prime \in C} \pcost_w(T^\prime, \{c^\prime\})\)
		\EndProcedure
	\end{algorithmic}
\end{algorithm}

We show the correctness and analyze the running time of \cref{algo:approx}.

\begin{theorem}
    \label{theo:approx}
    Given two parameters $\delta \in (0,1)$, $\epsilon \in (0, 1/2]$ and a set $T = \{\tau_1, \dots, \tau_n \} \subset \mathbb{X}^d_m$ of polygonal curves, with probability at least $1-\delta$ \cref{algo:approx} returns a $(5+\epsilon)$-approximate $(1,\ell)$-median for $T$ with $2\ell-2$ vertices, in time $$O\left(n m \log(m) + m^2 \log(m) \log^2(1/\delta) + m^{2\ell-1}\frac{\log(m)\log(n) + \log(1/\delta))\log(m)}{\epsilon^{2\ell d - 2d + 2}}\right).$$
\end{theorem}
\begin{proof}
    Let $c^\ast \in \argmin_{c \in \mathbb{X}^d_\ell} \cost{T}{\{c\}}$ be an optimal $(1,\ell)$-median for $T$. The expected distance between $s \in S$ and $c^\ast$ is $$\expected{d_F(s, c^\ast)} = \sum_{i=1}^n d_F(\tau_i, c^\ast) \cdot \frac{1}{n} = \frac{\cost{T}{\{c^\ast\}}}{n}.$$
 
    Now, using Markov's inequality, for every $s \in S$ we have $$ \Pr[d_F(s, c^\ast) > (1+\epsilon)\cost{T}{\{c^\ast\}}/n] \leq \frac{\cost{T}{\{c^\ast\}}n^{-1}}{(1+\epsilon)\cost{T}{\{c^\ast\}}n^{-1}} = \frac{1}{1+\epsilon},$$ therefore by independence $$ \Pr\left[\bigwedge_{s \in S} (d_F(s, c^\ast) > (1+\epsilon)\cost{T}{\{c^\ast\}}/n)\right] \leq \frac{1}{(1+\epsilon)^{\lvert S \rvert}} \leq \exp\left(-\frac{\epsilon \lvert S \rvert}{2}\right).$$ Hence, with probability at most $\exp\left(-\frac{\epsilon \left\lceil -\frac{2(\ln(\delta) - \ln(4))}{\epsilon} \right\rceil}{2}\right) \leq \delta/4$ there is no $s \in S$ with $d_F(s, c^\ast) \leq (1+\epsilon) \frac{\cost{T}{\{c^\ast\}}}{n}$. Now, assume there is a $s \in S$ with $d_F(s, c^\ast) \leq (1+\epsilon) \cost{T}{\{c^\ast\}}/n$. We do not want any $t \in S \setminus \{s\}$ with $\cost{T}{\{t\}} > (1+\epsilon) \cost{T}{\{s\}}$ to have $\cost{W}{\{t\}} \leq \cost{W}{\{s\}}$. Using \cref{theo:indyk_median}, we conclude that this happens with probability at most 
    \begin{align*}
        \exp\left(-\frac{\epsilon^2 \lceil -64\epsilon^{-2}(\ln(\delta) - \ln(\lceil -8 (\epsilon^\prime)^{-1} (\ln(\delta)-\ln(4))\rceil)) \rceil}{64}\right) \\
        \leq \frac{\delta}{\lceil -8 (\epsilon^\prime)^{-1} (\ln(\delta)-\ln(4))\rceil} \\
        \leq \frac{\delta}{4 \lvert S \rvert},
    \end{align*}
    for each $t \in S \setminus \{s\}$. Also, with probability at most $\delta/4$ \cite[Algorithm 1]{doi:10.1137/1.9781611976465.160} fails to compute a $34$-approximate $(1,\ell)$-median $\hat{c} \in \mathbb{X}^d_\ell$ for $T$, see \cite[Corollary 3.1]{doi:10.1137/1.9781611976465.160}, and with probability at most $\delta/4$, \cref{algo:coreset} fails to compute a weighted $\epsilon$-coreset for $T$, see \cref{theo:coreset}.
    \\
    
    Using a union bound over these bad events, we conclude that with probability at least $1-\delta$,
    \begin{itemize}
        \item \cref{algo:approx} samples a curve $s \in S$ with $d_F(s, c^\ast) \leq (1+\epsilon) \cost{T}{\{c^\ast\}}/n$,
        \item \cref{algo:approx} samples a curve $t \in S$ with $\cost{T}{\{t\}} \leq (1+\epsilon) \cost{T}{\{s\}}$,
        \item \cite[Algorithm 1]{doi:10.1137/1.9781611976465.160} computes a $34$-approximate $(1,\ell)$-median $\hat{c} \in \mathbb{X}^d_\ell$ for $T$, i.e., $\cost{T}{\{c^\ast\}} \leq \cost{T}{\{\hat{c}\}} \leq 34 \cost{T}{\{c^\ast\}}$
        \item and \cref{algo:coreset} computes a weighted $\epsilon$-coreset for $T$.
    \end{itemize}

    Using the triangle-inequality yields
    \begin{align*}
        \sum_{\tau \in T} (d_F(t, c^\ast) - d_F(\tau, c^\ast)) \leq \sum_{\tau \in T} d_F(t, \tau) & \leq (1+\epsilon) \sum_{\tau \in T} d_F(s, \tau) \\
        & \leq (1+\epsilon) \sum_{\tau \in T} (d_F(\tau, c^\ast) + d_F(c^\ast, s)),
    \end{align*}
    which is equivalent to
    \begin{align*}
        & n \cdot d_F(t, c^\ast) \leq (2+\epsilon) \cost{T}{c^\ast}  + (1+\epsilon) n  (1+\epsilon) \cost{T}{c^\ast}/n \\
        \Leftrightarrow{} & d_F(t, c^\ast) \leq (3 + 4\epsilon) \cost{T}{c^\ast}/n.
    \end{align*}
    Let $v^{t}_1, \dots, v^{t}_{\lvert t \rvert}$ be the vertices of $t$. By \cref{lem:shortcutting_single} there exists a polygonal curve $c^\prime \in \mathbb{X}^d_{2\ell-2}$ with every vertex contained in one of $B(v^{t}_1, d_F(c^\ast, t)), \dots, B(v^{t}_{\lvert t \rvert}, d_F(c^\ast, t))$ and $d_F(t, c^\prime) \leq d_F(t, c^\ast)$. We have $d_F(t, c^\prime) \leq d_F(t, c^\ast) \leq (3+4\epsilon) \cost{T}{\{c^\ast\}}/n$. Furthermore, by the $\epsilon$-coreset guarantee, see \cref{def:coreset}, we have \par $\lvert \Delta - \cost{T}{\{\hat{c}\}} \rvert \leq \epsilon \cost{T}{\{\hat{c}\}}$. Therefore, \par$\Delta_l = \Delta/(34 (1+\epsilon)) \leq \cost{T}{\{c^\ast\}} \leq \Delta_u = \Delta/(1-\epsilon)$ and $d_F(t, c^\prime) \leq (3+4\epsilon) \Delta_u/n$. We conclude that the set $C$ of all curves with up to $2\ell-2$ vertices from $P$, the union of the grid covers, contains a curve $c^{\prime \prime} \in \mathbb{X}^d_{2\ell-2}$ with distance at most $\frac{\epsilon\Delta_l}{n} \leq \epsilon\frac{\cost{T}{\{c^\ast\}}}{n}$ between every corresponding pair of vertices of $c^\prime$ and $c^{\prime\prime}$, thus $d_F(t, c^{\prime\prime}) \leq (3+5\epsilon) \cost{T}{\{c^\ast\}}/n$.
    
    In the last step, \cref{algo:approx} returns a curve $\overline{c} \in C$, that evaluates best against the $\epsilon$-coreset. By the $\epsilon$-coreset guarantee and the range of $\epsilon$, we know that $\cost{T}{\{\overline{c}\}} \leq (1+\epsilon)/(1-\epsilon) \cost{T}{\{c^{\prime\prime}\}} \leq (1+4\epsilon) \cost{T}{\{c^{\prime\prime}\}}$. We can now bound the cost of $\overline{c}$ as follows:
    \begin{align*}
        \cost{T}{\{\overline{c}\}} & \leq (1+4\epsilon) \sum_{\tau \in T} d_F(\tau, c^{\prime\prime}) \leq (1+4\epsilon) \sum_{\tau \in T} (d_F(\tau, t) + d_F(t, c^{\prime\prime})) \\
        & \leq (1+4\epsilon) \cost{T}{\{t\}} + (1+4\epsilon)(3+5\epsilon) \cost{T}{\{c^\ast\}} \\
        & \leq (1+\epsilon)(1+4\epsilon) \cost{T}{\{s\}} + (3+37\epsilon) \cost{T}{\{c^\ast\}} \\
        & \leq (1+9\epsilon) \sum_{\tau \in T}(d_F(\tau, c^\ast) + d_F(c^\ast, s)) + (3+37\epsilon) \cost{T}{\{c^\ast\}} \\
        & \leq (4+48\epsilon)\cost{T}{\{c^\ast\}} + (1+\epsilon)(1+9\epsilon) \cost{T}{\{c^\ast\}}  \\
        & \leq (5 + 67\epsilon) \cost{T}{\{c^\ast\}}
    \end{align*}
    Finally, we rescale $\epsilon$ by $\frac{1}{67}$ to obtain the desired approximation guarantee.
    \\
    
    We now discuss the running time. \cite[Algorithm 1]{doi:10.1137/1.9781611976465.160} has running time\par $O(m^2 \log(m) \log^2(1/\delta) + m^3 \log(m))$, see \cite[Corollary 3.1]{doi:10.1137/1.9781611976465.160} and \cref{algo:coreset} has running time $$O(n m \log(m) + m^2 \log(m) \log^2(1/\delta) + m^3 \log(m) + \epsilon^{-2}(\log(m)\log(n) + \log(1/\delta))),$$ see \cref{theo:coreset}. The $\epsilon$-coreset has size $O(\epsilon^{-2}(\log(m)\log(n) + \log(1/\delta)))$, therefore $\pcost_w(T^\prime, \hat{c})$ can be evaluated in time $O(m \epsilon^{-2} \log(m)(\log(m)\log(n) + \log(1/\delta)))$, using Alt and Godau's algorithm \cite{alt_godau} to compute the distances.
    
    The sample $S$ has size $O\left(\frac{-\ln(\delta)}{\epsilon}\right)$ and the sample $W$ has size $O\left(\frac{-\ln(\delta)}{\epsilon^2}\right)$. Evaluating each curve of $S$ against $W$ takes time $O\left(\frac{m^2 \log(m) \log^2(1/\delta)}{\epsilon^3}\right)$, using Alt and Godau's algorithm \cite{alt_godau} to compute the distances.
    
    Now, $c$ has up to $m$ vertices and every grid consists of \par$\left( \frac{\frac{2(3+4\epsilon^\prime)\Delta_u}{n}}{\frac{2\epsilon^\prime\Delta_l}{n\sqrt{d}}} \right)^d = \left( \frac{(3+4\epsilon^\prime)\sqrt{d}}{\epsilon^\prime} 34(1+\epsilon) \right)^d \in O\left(\frac{1}{\epsilon^d}\right)$ points (note that $\Delta_u/\Delta_l = (1+\epsilon^\prime)/(1-\epsilon^\prime) 34 \leq 34 (1+\epsilon)$). Therefore, we have $O\left(\frac{m}{\epsilon^{d}}\right)$ points in $P$ and \cref{algo:approx} enumerates all combinations of $2\ell-2$ points from $P$ taking time $O\left(\frac{m^{2\ell-2}}{\epsilon^{(2\ell-2)d}}\right)$. Afterwards, these candidates are evaluated against the $\epsilon$-coreset, which takes time $$O\left(\frac{m^{2\ell-1}(\log(m)\log(n) + \log(1/\delta))\log(m)}{\epsilon^{2\ell d-2d + 2}}\right),$$ using Alt and Godau's algorithm \cite{alt_godau} to compute the distances. All in all, we then have  running time $$O\left(n m \log(m) + m^2 \log(m) \log^2(1/\delta) + \frac{m^{2\ell-1}(\log(m)\log(n) + \log(1/\delta))\log(m)}{\epsilon^{2\ell d-2d + 2}}\right).$$
\qed\end{proof}

\section{Conclusion}
\label{sec:conclusion}

We presented an algorithm for computing $\epsilon$-coresets for $(k, \ell)$-median clustering of polygonal curves under the Fr\'echet distance and used these to improve the running time of an existing approximation algorithm for $(1,\ell)$-median clustering. %Our results have largely been possible due to a careful combination, combined with subtle changes, of previous works. At first glance, this may seem trivial, these results have nevertheless been unnoticed in the recent years.
Unfortunately, it was not possible to improve the existing $(k,\ell)$-median approximation algorithms in \cite{DBLP:conf/soda/DriemelKS16,doi:10.1137/1.9781611976465.160} by means of $\epsilon$-coresets. This is due to the recursive approximation scheme used in these works, where the candidate center sets are not necessarily evaluated against the input, but against subsets of the input. Thus, we would need an $\epsilon$-coreset for any subset of the input, which is not practical. We note that, to the best of our knowledge, no $(k,\ell)$-median clustering algorithm exists that do not employ this approximation scheme. 

It is still an interesting open problem whether there exist sublinear size $\epsilon$-coresets for weighted sets of polygonal curves. To derive such a result one may need a sublinear bound on the VC dimension of the range space of metric balls under scaled Fr\'echet distances, which is not evident at the moment. We note that such a result would enable the use of the iterative size reduction technique recently introduced by Braverman et al. \cite{DBLP:conf/soda/BravermanJKW21}.

\bibliographystyle{splncs03}
\bibliography{bibliography}

\end{document}